\documentclass[9pt,shortpaper,web,twoside]{ieeecolor}
\usepackage{generic}
\usepackage{cite}
\usepackage{amsmath,amssymb,amsfonts}
\usepackage{graphicx}
\usepackage{textcomp}
\usepackage{epstopdf}
\usepackage{diagbox}
\usepackage{booktabs}
\usepackage{multirow}
\usepackage{makecell}
\usepackage{subfig}
\usepackage{algorithm}
\usepackage{algorithmicx}
\usepackage[noend]{algpseudocode}

\def\BibTeX{{\rm B\kern-.05em{\sc i\kern-.025em b}\kern-.08em
    T\kern-.1667em\lower.7ex\hbox{E}\kern-.125emX}}
\DeclareMathOperator{\s}{\mathbf{s}}
\DeclareMathOperator{\x}{\mathbf{x}}
\DeclareMathOperator{\z}{\mathbf{z}}
\DeclareMathOperator{\aaa}{\mathbf{a}}
\DeclareMathOperator{\vv}{\mathbf{v}}
\DeclareMathOperator{\ii}{\textbf{i}}

\DeclareMathOperator{\col}{col}
\DeclareMathOperator{\diag}{diag}

\DeclareMathOperator{\BR}{BR}
\DeclareMathOperator{\I}{\mathbf{I}}
\DeclareMathOperator{\1}{\mathbf{1}}
\DeclareMathOperator{\0}{\mathbf{0}}
\newtheorem{thm}{Theorem}
\newtheorem{defi}{Definition}

\newtheorem{example}{Example}
\newtheorem{rmk}{Remark}
\newtheorem{asp}{Assumption}

\makeatletter
\def\fnum@figure{\textcolor{subsectioncolor}{\sf Fig.~\thefigure}}
\def\fnum@table{\textcolor{subsectioncolor}{\sf TABLE~\thetable}}
\def\thm@space@setup{%
  \thm@preskip=0.5pt
  \thm@postskip=\thm@preskip 
}
\makeatother
\begin{document}

\title{
Seeking Nash Equilibrium in Non-cooperative Quadratic Games Under Delayed Information Exchange}
\author{Kaichen Jiang, Yuyue Yan, \IEEEmembership{Member, IEEE}, Mingda Yue, and Yuhu Wu, \IEEEmembership{Member, IEEE}
\thanks{This work has been submitted to the IEEE for possible publication. Copyright may be transferred without notice, after which this version may no longer be accessible.
This work was supported in part by the National Natural Science Foundation of China under Grant U24A20263, 62595804. \textit{(Corresponding author: Yuhu Wu.)}}
\thanks{K. C. Jiang, M. D. Yue and Y. H. Wu are with the School of Control Science and Engineering, Dalian University of Technology, Dalian 116024, China (e-mail: jiangkc@mail.dlut.edu.cn; yuemingda@mail.dlut.edu.cn; wuyuhu@dlut.edu.cn).}
\thanks{Y. Yan is with the Department of Information Physics and Computing, The University of Tokyo, Bunkyo, Tokyo 113-8656, Japan (e-mail: yan-yuyue@g.ecc.u-tokyo.ac.jp).}
}

\maketitle
\begin{abstract}
  In this paper, we investigate the seeking of Nash equilibrium (NE) in a non-cooperative quadratic game where all agents exchange their delayed strategy information with their neighbors.
  To extend best-response algorithms to the delayed information setting, an estimation mechanism for each agent to estimate the current strategy profile is designed.
  Based on the best-response strategy to the estimations, the strategy profile dynamics of all agents is established, which is revealed to converge asymptotically to the NE when agents exchange multi-step-delay information via the Lyapunov-Krasovskii functional approach.
  In the scenario where agents exchange one-step-delay information, the exponential convergence of the strategy profile dynamics to the NE can be guaranteed by restricting the learning rate to less than an upper bound.
  Moreover, a lower bound on the learning rate for instability of the NE is proposed.
  Numerical simulations are provided for verifying the developed results.
\end{abstract}

\vspace{-2pt}
\section{Introduction}
\IEEEPARstart{N}{on-cooperative} games are used for dealing with the scenarios where agents behave selfishly with conflicting interests, such as human-machine interaction \cite{Varga2023}, traffic routing \cite{Varga2022}, cluster confrontation \cite{Liu2024}, etc.
As an important components of non-cooperative games, non-cooperative quadratic games (i.e., games with quadratic cost functions) are of significant interest due to their suitable roles in modeling various scenarios and practical applications \cite{Frihauf2012,Li2021,Martirosyan2024,Xu2024}.
Another reason, as discussed in Chapter 4.6 of \cite{Basar1997}, is that quadratic cost functions act as second-order approximations to various nonlinear cost functions, which are analytically tractable and admit closed-form equilibrium solutions.

Nash equilibrium (NE) is an attractive concept \cite{Nash1951} in non-cooperative games.
The existence of NE can be handled with many approaches \cite{Monderer1996,Liu2019}, however, the seeking of it is not easy.
In designing NE seeking algorithms, the strategy updating dynamics appears to be an important aspect since it allows agents to adjust strategies timely to cope with the information exchange until the strategy profile reaches a NE.
Generally speaking, the strategy updating dynamics can be divided into two types: continuous-time type \cite{Ye2017,Yan2024a,Yan2024b,Wang2025} and discrete-time type \cite{Marden2009,Belgioioso2023,Pavel2020}.
In this paper, we focus on discrete-time dynamics which can be easily implemented both in theoretical analysis and practical applications.

The information exchanged among agents plays an important role in designing strategy updating dynamics for each agent.
Many existing works concentrate on exchanging adequate strategy information among agents to construct valid strategy updating dynamics, such as fictitious play \cite{Shamma2005}, conditional imitation \cite{Govaert2021} and best-response dynamics \cite{Hart2006}.
However, when the number of agents becomes more larger and the communication graph becomes more complicated, the adequate strategy information exchange may not be admitted \cite{Lou2016,Fang2022,Xu2023}.
One approach for addressing this issue is to introduce agent's estimation of current strategy profile.
Through the estimation exchange, effective strategy updating dynamics can be established in either a semi-decentralized style \cite{Belgioioso2023} or a distributed style \cite{Ye2017,Huang2023}.

However, in many practical scenarios, both the strategy information and the estimation information being exchanged may be {outdated}, which is called \emph{delayed information exchange} in this paper.
On one hand, the delayed information may be caused by agents themselves since some of them may only be willing to disclose outdated information from earlier stages due to their selfishness and privacy concerns.
For example, in a dynamic bargaining, the buyer may hide the current strategy and  use former strategy to help achieve a negotiation with the seller \cite{Cramton1991,Spector2022}.
On the other hand, the communication delay phenomenon also gives rise to the delayed information, which is quite common in distributed networked systems \cite{Wang2022,Ai2020,LiuJ2024,Li2025}.

Taking the delayed information into consideration brings great challenges on both the design and the convergence analysis of NE seeking algorithms.
Some scholars addressed the effect of communication delay by embedding virtual agents into the original communication graph, where the amount of virtual agents is the sum  of all agents' delays \cite{Nedic2010,LiuJ2024,Chen2024}.
This approach run the risk of enlarging the dimensionality of adjacency matrix of the reformulated
communication graph.
Another widely adopted approach involves leveraging Lyapunov-Krasovskii functional theory to establish the convergence condition for the designed algorithm \cite{Ai2020,Seuret2015}.
However, the construction of an apt Lyapunov-Krasovskii functional candidate remains challenging.
When addressing more complex scenarios such as agents with continuous dynamics and model uncertainties \cite{Wang2022}, the algorithm design and convergence analysis  become significantly   challenging.


Based on the above discussions, in this paper, we consider a special class of decision dynamics for non-cooperative quadratic games in which each agent only exchanges the delayed strategy and estimation information with his neighbors to hide sensitive information.
We aim to design an estimation update mechanism that incorporates only delayed information, enabling agents to estimate the current strategy profile while helping protect their current strategies from disclosure.
Combining the best-response strategy to the updated estimation, the strategy updating dynamics are established for agent to achieve the convergence of the NE of the quadratic game.
After analyzing the influence of the delayed information on the designed seeking dynamics, we obtain two interesting results for different cases of delay.
It turns out that exchanging information with longer delays among agents necessitates reducing the learning rate to ensure the convergence of the strategy update dynamics.
The main contributions in this paper are twofold:
\begin{itemize}
  \item For the case where the agents exchange $\tau$-step-delay information with $\tau\geq2$, we prove 
      that using an appropriate learning rate in the designed estimation updating mechanism, the established strategy updating dynamics is ensured to asymptotically converge to the NE via the Lyapunov-Krasovskii functional approach.
  \item Different from the asymptotical converge result obtained for $\tau\geq2$, when agents exchange one-step-delay information (i.e., $\tau=1$), we prove 
      that the established strategy updating dynamics can exponentially converge to the NE by using a proper learning rate less than an upper bound. 
      Moreover, we prove 
      that the strategy updating dynamics is divergent when the learning rate exceeds a lower bound.
\end{itemize}

The rest of the paper is organized as follows.
In Section II, the problem formulation on non-cooperative quadratic game model and best-response-based strategy profile dynamics are provided.
In Section III, agent's estimation dynamics and the corresponding strategy profile dynamics are established.
Sufficient conditions to guarantee the convergence of the NE are proposed for different steps of delays, along with numerical examples.
Section IV is the concluding remarks.

Notations: We write $\mathbb{R}$ (resp. $\mathbb{C}$) and $\mathbb{N}$ (resp. $\mathbb{N}_+$) for the set of real (resp. complex) and non-negative (resp. positive) natural numbers.
We write $\mathbb{R}^n$, $\mathbb{R}^{n\times m}$ and $\mathbb{S}_+^n$ for the space of $n$-dimensional real column vectors, $n\times m$ real matrices and $n$-dimensional symmetric positive definite matrices.
Given a matrix $A\in\mathbb{R}^{n\times n}$, $[A]_{ij}$ denotes the $i$th row $j$th column element of $A$ 
and $\rho(A)$ denotes the spectral radius of $A$.
$\mathbf{I}_n$ denotes the $n$-dimensional identity matrix.
$\mathbf{1}_n$ (resp. $\mathbf{0}_n$) denotes the $n$-dimensional all one (resp. all zero) column vector.
For notation ease, we use $\mathbf{I}$ and $\mathbf{0}$ to represent the identity matrix and zero matrix with required dimensions when there is no ambiguity.
$\top$ is the transpose and $\otimes$ is the Kronecker product.
Given $p$ vectors $\mathbf{v}_i\in\mathbb{R}^n$, $i=1,\ldots,p$, define $\col(\mathbf{v}_1,\ldots,\mathbf{v}_p)=(\mathbf{v}_1^\top,\ldots,\mathbf{v}_p^\top)^\top$ as a stacking column vector.
Given $p$ real matrices $M_i\in\mathbb{R}^{n\times m}$, $i=1,\ldots,p$, define $\diag(M_1,\ldots,M_p)$ as a block diagonal matrix whose diagonal blocks are $M_1,\ldots,M_p$.

An undirected graph is denoted by $(\mathcal{N},\mathcal{E})$ with $\mathcal{N}$ as the set of nodes and $\mathcal{E}$ as the set of edges.
The adjacent matrix of $(\mathcal{N},\mathcal{E})$ is denoted by an $n$-dimensional symmetric matrix $\mathcal{A}$ whose leading diagonal elements are zero and other elements $\mathcal{A}_{ij}=\alpha_{ij}$ satisfy $\alpha_{ij}=1$ if there exists an edge links nodes pair $(i,j)$ and $\alpha_{ij}=0$ otherwise.
The set of neighbors of  node $i$ is denoted by $\mathcal{N}_i\triangleq\{j\in\mathcal{N}: \alpha_{ij}=1\}$ and the Laplacian matrix of  $(\mathcal{N},\mathcal{E})$ is denoted by $\mathcal{L}$.

\section{Problem Formulation}

\subsection{Static Quadratic Game}
{Consider a non-cooperative quadratic game with the set of agents denoted by $\mathcal{N}\!=\!\{1,2,\cdots,n\}$ where agent $i$'s strategy  is denoted by $s_i\in \mathbb{R}$.
The strategy profile of all agents is denoted by a column vector $\mathbf{s}\triangleq(s_1,\ldots,s_n)^\top=(s_i,\mathbf{s}_{-i})^\top\in\mathbb{R}^n$, in which $\mathbf{s}_{-i}$ is the strategy profile excluding $i$.
Similar to \cite{Frihauf2012,Yan2022}, agent $i$'s quadratic payoff function  $J_i: \mathbb{R}^n\rightarrow\mathbb{R}$ is only available to $i$ and is given by
\begin{align}\label{CF}
J_i(s_i,\mathbf{s}_{-i})=\frac{1}{2}\s^{\top}A_i \s+\mathbf{b}_i^{\top} \s+g_i,
\end{align}
where $A_i=[a_{kj}^i]_{k,j\in\mathcal{N}}\in \mathbb{R}^{n\times n}$ is a symmetric  matrix satisfying $a_{ii}^i<0$, $\mathbf{b}_i=(b_1^i,\cdots,b_n^i)^\top\in\mathbb{R}^n$, and $g_i\in\mathbb{R}$.
The quadratic game investigated in this paper is denoted by $G=\{\mathcal{N},$ $(s_i)_{i\in\mathcal{N}},(J_i)_{i\in\mathcal{N}}\}$, where agents  aim to maximize their payoffs.

In noncooperative games, NE is a strategy profile where no agent has incentive to unilaterally change his strategy for earning more payoff.
The definition is given in the following.
\begin{defi}\label{NEdefi}
A strategy profile $\s^*=(s^{*}_i,\s^{*}_{-i})\in\mathbb{R}^n$ in  quadratic game $G$ is called a (pure-strategy) NE if for any agent $i\in\mathcal{N}$,
\begin{align}
J_i(s^{*}_i,\s^{*}_{-i})\geq J_i(s_i,\s^{*}_{-i}),\ \forall s_i\in S_i.
\end{align}
\end{defi}

Here, we adopt the following assumption utilized in \cite{Frihauf2012,Ye2017}, to ensure the existence and uniqueness of NE in the quadratic game $G$.
\begin{asp}\label{strctilydominant}
The elements $a_{ij}^i$, $j=1,\ldots,n$ in $A_i$ satisfy
\begin{align}\label{basecondition}
|a_{ii}^i|>\sum\limits_{j\in\mathcal{N},j\neq i}|a_{ij}^i|,\  \forall i\in\mathcal{N}.
\end{align}
\end{asp}
\subsection{Myopic Best-Response-Based Seeking Dynamics}
In this paper, we aim to design a distributed algorithm for seeking the NE of the quadratic game $G$ by utilizing the information exchanged among agents through an undirected graph $(\mathcal{N},\mathcal{E})$, which is assumed to be connected.

A commonly used seeking dynamics is the so-called myopic best-response-based dynamics \cite{Liu2019}, which means each agent updates his strategy according to the strategy profile at previous stage.
Let $s_i(t)$ denote agent $i$'s strategy at stage $t$, $t\in\mathbb{N}$, then agent $i$ updates his strategy $s_i(t+1)$ by
\begin{align}\label{Complete}
s_i(t+1)\in\BR_i(\s_{-i}),
\end{align}
where $\BR_i(\s_{-i})$ denotes the set of best-response strategies of agent $i\in\mathcal{N}$ that maximizes his payoff function $J_i(s_i,\s_{-i})$ given other agents' strategies $\s_{-i}$, which is defined by
\begin{align}\nonumber
\BR_i(\s_{-i})=&\arg\max\limits_{s'_i\in \mathbb{R}} \Big(\frac{1}{2}(s'_i,\s_{-i})^{\!\top}A_i (s'_i,\s_{-i})\!+\!\mathbf{b}_i^{\!\top}(s'_i,\s_{-i})\Big)\\\label{BRdefi}
=&-\frac{1}{a_{ii}^i}{\big(\!\sum\limits_{j\neq i,j\in\mathcal{N}}\!a_{ij}^i s_j+b_i^i\big)}.
\end{align}
Note that $\BR_i(\s_{-i})$ is a singleton in this paper since it is uniquely determined by strategies $\s_{-i}$ according to \eqref{BRdefi}.
Thus, the myopic best-response-based strategy profile dynamics 
can be represented by the following compact form
\begin{align}\label{QFsys}
\s(t+1)=M\s(t)+\mathbf{c},
\end{align}
where $\mathbf{c}=-({b_{1}^1}/{a_{11}^1},\ldots,{{b_{n}^n}}/{a_{nn}^n})^\top$ and $M=[\aaa_1,\ldots,\aaa_n]^\top$ with $\aaa_i\triangleq$ $-\col\big({ a_{i1}^i}/{a_{ii}^i},\ldots,
{a_{ii-1}^i}/{a_{ii}^i},0,{a_{ii+1}^i}/{a_{ii}^i},\ldots,{a_{in}^i}/{a_{ii}^i}\big)\in\mathbb{R}^n$ for $i\in\mathcal{N}$.

Similar to \cite{Frihauf2012,Ye2017}, it can be verified that under Assumption~\ref{strctilydominant}, the NE of the quadratic game $G$ is represented by $(\I-M)^{-1}\mathbf{c}$, which is exactly the equilibrium state of \eqref{QFsys}.
\begin{rmk}
Note that none of agents knows the matrix $M$ and $\mathbf{c}$ (hence they neither know the NE $\s^*$) since agents only know their individual payoff functions, i.e., $A_i$ and $\mathbf{b}_i$ are known to agent $i$ only.
\end{rmk}

It is worth noting that the dynamics \eqref{QFsys} require the full observation of the other agents strategy $s_{-i}$ at current stage $t$, which may not be allowed in many practical scenarios
\cite{Lou2016,Fang2022,Xu2023}.
Moreover, in many cases, the information is delayed due to some communication issue \cite{LiuJ2024} or information disclosure problem \cite{Cramton1991}.
For example, agents may only disclose $s_i(t-\tau)$ at stage $t$ to keep privacy of his recent decision strategy.
Therefore, in this paper, we are interested to design the NE seeking algorithm with delayed information exchange.

\section{Design of Seeking Algorithm for NE}
\subsection{Design of Seeking Algorithm}

Similar to \cite{Lou2016,Fang2022} where the current strategies $\s_{-i}(t)$ of the others is unknown to agent $i$, we deal with the problem by introducing agent $i$'s estimation $\hat{\s}^i(t)$ on $\s(t)$, which is composed of the estimations $\hat{s}_{ij}(t)$ of all other agent $j$'s current strategies $s_j(t)$, i.e.,
\begin{align}\label{estimationofj}
\hat{\s}^{i}(t)=(\hat{s}_{i1}(t),\cdots, \hat{s}_{in}(t))^{\top}\in\mathbb{R}^n,\ i\in\mathcal{N},
\end{align}
so that agent $i$ could update his strategy $s_i(t+1)$ by adopting the best-response strategy to his estimation $\hat{\s}^{i}(t)$ under
\begin{align}\label{strategydyna}
s_i(t+1)=-\frac{1}{a_{ii}^i}{\big(\!\sum\limits_{j\neq i,j\in\mathcal{N}}\!a_{ij}^i \hat{s}_{ij}(t)+b_i^i\big)}.
\end{align}
In this case, letting all agents' estimations on   $\s(t)$ be stacked by
\begin{align}\label{Allestimations}
\hat{\s}(t)=\col(\hat{\s}^1(t),\ldots,\hat{\s}^n(t))\in\mathbb{R}^{n^2}.
\end{align}
Similar to \eqref{QFsys}, from  agents' strategy updating dynamics \eqref{strategydyna}, one establishes the best-response-based strategy profile dynamics as follows
\begin{align}\label{brdynamic}
\s(t+1)=N\hat{\s}(t)+\mathbf{c}, \ t\in\mathbb{N},
\end{align}
where $\mathbf{c}$ is defined in \eqref{QFsys} and $N\in\mathbb{R}^{n\times n^2}$  is given by
\begin{align}\nonumber
N=\diag\Big(&\big(0,-\frac{a_{12}^1}{a_{11}^1},\cdots,-\frac{a_{1n}^1}{a_{11}^1}\big), \big(-\frac{a_{21}^2}{a_{22}^2}, 0,\cdots,-\frac{a_{2n}^2}{a_{22}^2}\big),\\
&\ldots,(-\frac{a_{n1}^n}{a_{nn}^n},-\frac{a_{n2}^n}{a_{nn}^n},\cdots,0)\Big).
\end{align}

Next, we aim to design an updating dynamics for agent $i$'s estimation $\hat{\s}^i(t)$ of current strategy profile $\s(t)$ to achieve the convergence of \eqref{brdynamic} to the NE.
In this paper, we consider the scenario where each agent $i\in\mathcal{N}$  can only receive his neighbor $j$'s $\tau$-step-delay strategy $s_j(t-\tau)$ and estimation $\hat{\s}^j(t-\tau)$ at current stage $t$.
Different from \cite{Lou2016,Fang2022} where the information exchanged therein is from current stage, the design of estimation should be properly handled with respect to the delayed information in this paper.
Using the above delayed information, we propose the following updating dynamics for agent $i$'s estimation $\hat{\s}^{i}(t)$ defined in \eqref{estimationofj} given by
\begin{align}\nonumber\label{estiupdate3}
\hat{s}_{ij}(t)=&\hat{s}_{ij}(t-1)+\xi\Big[\sum\limits_{k\in \mathcal{N}_i} \big(\hat{s}_{kj}(t-\tau)-\hat{s}_{ij}(t-\tau)\big)\\
& +\alpha_{ij}(s_j(t-\tau)-\hat{s}_{ij}(t-\tau))\Big],\ t\in\mathbb{N},\ i,j\in\mathcal{N},
\end{align}
where  $\xi>0$ is the learning rate controlling how large of a step to take in the direction of the difference between agent $i$'s own estimation and  the delayed estimation from his neighbors, $\tau\in\mathbb{N}_+$ denotes the delay step, 
$\hat{s}_{ij}(k)=\phi_{ij}(k)$ for $k=-\tau,\ldots,-1$ and ${\s}(k)=\varphi(k)$ for $k=-\tau,\ldots,0$ with $\phi_{ij}$ and $\varphi$ being the initial conditions.
In \eqref{estiupdate3}, the difference $\hat{s}_{kj}(t-\tau)-\hat{s}_{ij}(t-\tau)$ is used to achieve the consensus among agents and the difference $s_j(t-\tau)-\hat{s}_{ij}(t-\tau)$ is used to reduce the bias between the estimations and true strategies.

\begin{algorithm}[h]
\caption{Seeking algorithm with delayed information exchange}
\label{SeekingAlgo}
\begin{algorithmic}[1]
\State Initialization: $\hat{s}_{ij}(k)$ and $\s(k)$ for $k=-\tau,\ldots,-1$, $\tau\in\mathbb{N}_+$, $\xi>0$, $t=-1$.
\State Iteration: $t\leftarrow t+1$
\State Estimation update:
\For{$i=1\rightarrow n$}
\For{$j=1\rightarrow n$}
        \eqref{estiupdate3}
    \EndFor
\EndFor
\State Strategy update:
\For{$i=1\rightarrow n$}
        \eqref{strategydyna}
\EndFor
\end{algorithmic}
\end{algorithm}
\vspace{-1pt}
Along with the designed updating dynamics \eqref{estiupdate3} of agent's estimation, the NE seeking algorithm with delayed information exchange is established in this paper as shown in Algorithm~\ref{SeekingAlgo}.
Unlike 
\cite{Pavel2020}, it can be observed from \eqref{brdynamic} that the update of strategy profile $\s(t+1)$ in this paper does not involve the information of strategy profile ${\s}(t)$.
Combining \eqref{strategydyna} with \eqref{estiupdate3}, the update of agent $i$'s strategy $s_i(t+1)$  depends on the strategy profile $\s(t-1)$ and estimation $\hat{\s}(t-1)$  when $\tau=1$, implying that updating strategy profile $\s(t+1)$ requires one-step-delay information.
When $\tau\geq2$, the manifestation of delay in updating  $\s(t+1)$ is more pronounced.
Noting that the  estimation updating dynamics \eqref{estiupdate3} degenerates to a delay-free case when $\tau=1$, which may lead to distinct strategy profile dynamics  compared with the case of $\tau\geq2$.
Next, we respectively investigate the influence of two cases of delay (i.e., $\tau \geq2$ and $\tau = 1$) on the information exchange and provide different conditions to guarantee the convergence of Algorithm~1 to the NE of the quadratic game $G$.

\subsection{Convergence Analysis of Algorithm~\ref{SeekingAlgo} With $\tau\geq2$}
In this subsection, we consider the first case where agents exchange $\tau$-step-delay information with their neighbors and $\tau\geq2$.
Denote $\bar{\s}\triangleq{\bf 1}_n\otimes \s$.
In this case, the designed estimation updating dynamics \eqref{estiupdate3} of all agents can be represented in the following compact form,
\begin{align}\label{estsystem3}
\hat{\s}(t)=\hat{\s}(t-1)+\xi P\hat{\s}(t-\tau)+\xi B\bar{\s}(t-\tau),
\end{align}
where $P$ and $B\in\mathbb{R}^{n^2\times n^2}$ are defined by
\begin{align}\label{Pdefi}
&P\triangleq-(\mathcal{L}\otimes \mathbf{I}_{n}+B),\\\label{Bdefi}
&B\triangleq\diag(\alpha_{11},\ldots,\alpha_{1n},\ldots,\alpha_{n1},\ldots,\alpha_{nn}),
\end{align}
with $\mathcal{L}$ being the Laplacian matrix of communication graph $(\mathcal{N},\mathcal{E})$.

Define $\mathfrak{s}(t)\triangleq\col(\bar{\s}(t),\hat{\s}(t))\in\mathbb{R}^{2n^2}$, which is composed of the strategy profile $\s(t)$ at stage $t$ and agents' estimations $\hat{\s}^i(t)$ of $\s(t)$.
To handle the convergence analysis of the designed seeking dynamics \eqref{strategydyna} along with the estimation dynamics \eqref{estiupdate3}, we combine \eqref{brdynamic} with \eqref{estsystem3} and construct an auxiliary system in the following,
\begin{align}\label{combinedsys}
    \mathfrak{s}(t)=H_1\mathfrak{s}(t-1)+H_2(\xi)\mathfrak{s}(t-\tau)+\bar{\mathbf{c}},\ t\in\mathbb{N},
\end{align}
where $\bar{\mathbf{c}}=\col({\bf 1}_n\!\otimes\mathbf{c}, {\bf 0}_{n^2})$, $H_1,H_2(\xi)\!\in\!\mathbb{R}^{2n^2\times2n^2}$ are defined by
\begin{align}\setlength{\arraycolsep}{1.4pt}\!\!\label{H1H2}
 H_1\!=\!\left[ {\begin{array}{cc}  \0_{n^2\times n^2} &  \mathbf{1}_n\otimes N  \\
\0_{n^2\times n^2} & \I_{n^2}\\
\end{array} } \right]\!,
H_2(\xi)\!=\!\xi\left[ {\begin{array}{cc}  {\0_{n^2\times n^2}} & \0_{n^2\times n^2} \\
B & P\\
\end{array} } \right]\!.
\end{align}
For ease of notation, let $H_2\triangleq H_2(\xi)$ hereinafter.
To analyze the convergence of system \eqref{combinedsys}, one constructs the following system
\begin{align}\label{transferedsys}
    \x(t)=H_1\x(t-1)+H_2\x(t-\tau),\ t\in\mathbb{N}_+,
\end{align}
where $\mathbf{x}(t)\triangleq\mathfrak{s}(t)-\mathfrak{s}(t-1)$ and $\x(k)=\psi(k)$ for $k=-\tau+1,\ldots,0$ with $\psi$ being the initial condition.
Next, we adopt the Lyapunov-Krasovskii functional theory to handle the convergence analysis of   \eqref{transferedsys} and then provide sufficient condition for a given learning rate $\xi$ and a given delay step $\tau$ to achieve the convergence of Algorithm~\ref{SeekingAlgo}.

\begin{thm}\label{thm1}
    Suppose that Assumption \ref{strctilydominant} is satisfied and all agents exchange $\tau$-step-delay ($\tau\geq2$) information at current stage $t$.
    Under Algorithm~\ref{SeekingAlgo} with  learning rate $\xi>0$,
    if there exist matrices $Q_1\in\mathbb{S}_+^{4n^2}$ and $Q_2,Q_3\in\mathbb{S}_+^{2n^2}$  such that
    \begin{align}\label{condi3}
        F(\tau,\xi)\prec 0,
    \end{align}
    where $F(\tau,\xi)=E_2^\top(\xi) Q_1E_2(\xi)-E_1^\top Q_1E_1+\diag(Q_2,-Q_2,{\bf0})+$ $D\!^\top\!\!(\tau,\xi)\diag((\tau\!-\!1) Q_3,\frac{1}{1-\tau} Q_3,\frac{-3\tau}{\tau^2-3\tau+2}Q_3) D(\tau,\xi)\!\in\!\mathbb{R}^{6n^2\!\times6n^2}$,
    \begin{align*}\setlength{\arraycolsep}{1pt}
        E_1^\top\!=\!\left[ {\begin{array}{cc} \I & -\I  \\
        {\bf0} &{\bf0}\\
     {\bf0}& \I \\
\end{array} } \right]\!\!,
        E_2^\top\!(\xi)\!=\!\left[ {\begin{array}{cc} H_1^\top & {\bf0}\\H_2^\top & -\I\\{\bf0} & \I  \\
\end{array} } \right]\!\!,
    D(\tau,\xi)\!=\!\left[ {\begin{array}{ccc} H_1\!\!-\!\I & H_2 & {\bf0} \\
    \I & -\I & {\bf0} \\ \I &  \I & -\frac{2}{\tau}\I \\
\end{array} } \right]\!\!,
\end{align*}
then the strategy profile dynamics of all agents  globally asymptotic- ally converge to the NE $\s^*$ of the quadratic game $G$.
\end{thm}

{\it Proof}. For verifying the asymptotic stability of system \eqref{transferedsys}, we construct the following Lyapunov-Krasovskii functional  candidate $V(\x(t))=V_1(\x(t))+V_2(\x(t))+V_3(\x(t))$,
where
\begin{align}\label{LKF1}\setlength{\arraycolsep}{1pt}
    &V_1(\x(t))= \chi(t)^\top Q_1 \chi(t),\\\label{LKF2}
    &V_2(\x(t))=\sum\limits_{i=t-\tau+1}^{t-1}\x^\top(i)Q_2\x(i),\\\label{LKF3}
    &V_3(\x(t))=\sum\limits_{i=t-\tau+2}^{t}(i-t+\tau-1)\z^\top(i)Q_3\z(i),
\end{align}
with $\chi(t)=\col\!\big(\x(t),\sum_{i=t-\tau+1}^{t-1}\x(i)\big)$, $\z(i)=\x(i)-\x(i-1)$.

To show the forward difference of $V(\x(t))$, we introduce the following augmented state
$\mathbf{v}(t)=\col(\x(t),\x(t-\tau+1),\z_1(t))$
with $\z_1(t)\triangleq\sum_{i=t-\tau+1}^{t}\x(i)$.
Then via \eqref{transferedsys} and $\z_1(t+1)-\x(t+1)=\z_1(t)-\x(t-\tau+1)$, one can rewrite $\chi(t)$ and $\chi(t+1)$  by
\begin{align}\!\label{v1t} \setlength{\arraycolsep}{2.0pt}
    \chi(t)\!=\!\left[ {\begin{array}{c}  \x(t) \\
    \z_1(t)-\x(t) \\
\end{array} } \right]
    \!=\!\left[ {\begin{array}{ccc} \I & {\bf0} & {\bf0}  \\
    -\I & {\bf0} &\I \\
\end{array} } \right]\mathbf{v}(t)\!=\! E_1  \mathbf{v}(t),
\end{align}
and
\begin{align}\!\label{v1t1}
\chi(t+1)  = \begin{bmatrix}
H_1 \mathbf{x}(t) + H_2 \mathbf{x}(t-\tau+1) \\ \mathbf{z}_1(t) - \mathbf{x}(t-\tau+1)
\end{bmatrix}\! = \!\underbrace{\begin{bmatrix}\setlength{\arraycolsep}{1.0pt}
 H_1 & \!\!\!\!\!\!H_2 & \!\!\!\!\mathbf{0} \\ \mathbf{0} & \!\!\!\!\!\!-\mathbf{I} & \!\!\!\!\mathbf{I}
 \end{bmatrix}}_{E_2(\xi)}\mathbf{v}(t).
\end{align}
With \eqref{v1t} and \eqref{v1t1}, one has
\begin{align}\nonumber
    \Delta V_1(t)&=V_1(\x(t+1))-V_1(\x(t))\\\label{deltaV1}
    &=\mathbf{v}^\top(t)\underbrace{\big(E_2^\top(\xi) Q_1E_2(\xi)-E_1^\top Q_1E_1\big)}_{F_1(\xi)} \mathbf{v}(t).
\end{align}

According to the function $V_2$ in \eqref{LKF2}, one has
\begin{align}\nonumber
    \Delta V_2(t)&=V_2(\x(t+1))-V_2(\x(t))\\\nonumber
    &=\sum\limits_{i=t-\tau+2}^{t}\x^\top(i)Q_2\x(i)-\sum\limits_{i=t-\tau+1}^{t-1}\x^\top(i)Q_2\x(i)\\\nonumber
    &=\x^\top(t)Q_2\x(t)-\x^\top(t-\tau+1)Q_2\x(t-\tau+1)\\\label{deltaV2}
    &=\mathbf{v}^\top(t)\diag(Q_2,-Q_2,{\bf0}) \mathbf{v}(t)\triangleq \mathbf{v}^\top(t)F_2\mathbf{v}(t).
\end{align}

Finally, for the function $V_3$ in \eqref{LKF3}, one can prove that
\begin{align*}
    \Delta V_3(t)&=V_3(\x(t+1))-V_3(\x(t))\\
    &=\sum\limits_{i=t-\tau+3}^{t+1}(i-t+\tau-2)\z^\top(i)Q_3\z(i)\\
    &\quad-\sum\limits_{i=t-\tau+2}^{t}(i-t+\tau-1)\z^\top(i)Q_3\z(i)\\
    &=(\tau-1) \z^\top(t+1)Q_3\z(t+1)-\!\!\!\!\sum\limits_{i=t-\tau+2}^{t}\!\!\!\z^\top(i)Q_3\z(i).
\end{align*}
It directly follows from Lemma 2 of \cite{Seuret2015} that the positive definiteness of matrix $Q_3$ implies $-\sum_{i=t-\tau+2}^{t}$ $\z^\top(i)Q_3\z(i)$ is bounded, i.e.,
\begin{align}\!\!\label{lemma2}\setlength{\arraycolsep}{1.0pt}
    -\!\!\!\!\sum\limits_{i=t-\tau+2}^{t}\!\!\!\z^\top(i)Q_3\z(i)\!\leq\! -\frac{1}{\tau-1}\!\left[ {\begin{array}{c} \phi_1 \\
    \phi_2 \\
\end{array} } \right]^\top\!\!\left[ {\begin{array}{cc} Q_3 & {\bf 0}\\
    {\bf 0} &\frac{3\tau}{\tau-2}Q_3\\
\end{array} } \right]\!
\left[ {\begin{array}{c} \phi_1 \\
    \phi_2 \\
\end{array} } \right],
\end{align}
where $\phi_1=\x(t)-\x(t-\tau+1)$, $\phi_2=\x(t)+\x(t-\tau+1)-\frac{2}{\tau}\z_1(t)$.

Notice that $\z(t+1)=\x(t+1)-\x(t)=(H_1-I)\x(t)+H_2\x(t-\tau+1)$, combined with \eqref{lemma2}, one further has
\begin{align}\setlength{\arraycolsep}{2.0pt}\nonumber
    \Delta V_3(t)\!\leq\!&~\mathbf{v}^\top(t)D^\top(\tau,\xi) \diag((\tau-1) Q_3,\frac{1}{1-\tau} Q_3,\\ \!\!\label{deltaV3}
    &\frac{-3\tau}{\tau^2-3\tau+2}Q_3)D(\tau,\xi) \mathbf{v}(t) \!\triangleq\! \mathbf{v}^\top(t)F_3(\tau,\xi) \mathbf{v}(t).
\end{align}

Combining the forward differences \eqref{deltaV1}-\eqref{deltaV3}, we arrive at
\begin{align*}
    \Delta &V(t)=\Delta V_1(t)+\Delta V_2(t)+\Delta V_3(t)
    \\
    &\leq \mathbf{v}^\top(t) [F_1(\xi) +F_2+F_3(\tau,\xi)] \mathbf{v}(t)=\mathbf{v}^\top(t)F(\tau,\xi)  \mathbf{v}(t)<0.
\end{align*}
Thus the asymptotical convergence of the system \eqref{transferedsys} is guaranteed.
Denote the equilibrium state of system \eqref{transferedsys} by $\x^e$, then
\begin{align}
    (\I-H_1-H_2)\x^e={\bf 0}_{2^{n^2}}.
\end{align}
Premultiplying both sides of the above equation by a  nonsingular matrix  {\footnotesize $  {\begin{bmatrix} \I & {\bf 0}\\
 -{\xi} B &  \I \\
\end{bmatrix} }  $},
we obtain
\begin{align}\label{NonsingularMatrix}\setlength{\arraycolsep}{2.0pt}
\left[ {\begin{array}{cc}  \I  & - {\1}_n\otimes N\\
 {\bf 0} &  -\xi B({\1}_n\otimes N)+\xi \mathcal{L}\otimes \I+\xi B \\
\end{array} } \right]\x^e= {\bf 0}_{2^{n^2}}.
\end{align}
Since the matrix $B({\bf 1}_n\otimes N-\I_{n^2})-\mathcal{L}\otimes \I_n$ is nonsingular\footnote{Notice that the matrix $B({\bf 1}_n\otimes N-\I_{n^2})-\mathcal{L}\otimes \I_n$ is irreducibly diagonally dominant, thus it is nonsingular according to Corollary 6.2.27 in \cite{Horn2012}.}, one has $\x^e={\0}_{2^{n^2}}$, that is,  the equilibrium state of system \eqref{transferedsys} is only the zero state.
Note that $\x(t)=\mathfrak{s}(t)-\mathfrak{s}(t-1)$, then system \eqref{combinedsys}  can be ensured to converge.
Next we prove that the equilibrium state $\mathfrak{s}^e=\col(\bar{\s}^e,\hat{\s}^e)$ of system \eqref{combinedsys} is exactly the augmented state composed of NE, i.e., $\bar{\s}^e=\hat{\s}^e={\bf 1}_{n}\otimes \s^*$.
From \eqref{combinedsys}, one has
\begin{align}\label{equilibriumstate}
\mathfrak{s}^e=H_1\mathfrak{s}^e+H_2\mathfrak{s}^e+\bar{\mathbf{c}}=(H_1+H_2) \mathfrak{s}^e+\bar{\mathbf{c}},
\end{align}

Recall \eqref{estiupdate3}, at the equilibrium state $\mathfrak{s}^e$, one has
\begin{align}\nonumber
\sum\limits_{k\in \mathcal{N}_i} \big(\hat{s}_{kj}^e-\hat{s}^e_{kj}\big) +\alpha_{ij}(s_j^e-\hat{s}_{ij}^e)=0,\ \forall i,j\in\mathcal{N}.
\end{align}
Then, $\hat{s}_{ij}^e={s}_{j}^e$ holds for any $i,j\in\mathcal{N}$ as the communication graph is undirected and connected \cite{Ye2017}.
As a result, one has $\bar{\s}^e=\hat{\s}^e$.

Furthermore, \eqref{equilibriumstate} implies  $\bar{\s}^e=({\bf 1}_n\otimes N)\hat{\s}^e={\bf 1}_n\otimes \mathbf{c}$, hence
\begin{align*}
    \s^e=N\hat{\s}^e+\mathbf{c}=N\bar{\s}^e+\mathbf{c}=N({\bf 1}_n\otimes {\s}^e)+\mathbf{c}=M{\s}^e+\mathbf{c},
\end{align*}
where $M$ is defined in \eqref{QFsys}.
Note that the matrix ${\bf I}-M$ is nonsingular by applying condition \eqref{basecondition} of Assumption \ref{strctilydominant},  one solves $\s^e=({\bf I}-M)^{-1}\mathbf{c}$, which  is exactly the NE $\s^*$ of the {quadratic} game $G$.
$\hfill\Box$

\begin{rmk}
It should be emphasized that the Lyapunov-Krasovskii functional candidate in the proof of Theorem~\ref{thm1} is only applicable for $\tau\geq2$.
When $\tau=1$, the lower limits of the summations in $V_2(\x(t))$ and $V_3(\x(t))$ in \eqref{LKF2}--\eqref{LKF3} are greater than their upper limits,
which leads to the undefined expressions, making it impossible to conduct proper analysis.
The convergence analysis for $\tau=1$ will be presented in Section III-C later.
\end{rmk}

\begin{rmk}
The structural complexity and high dimensionality of matrix $F(\tau,\xi)$ present significant challenges for theoretically determining a learning rate $\xi$ that ensures the convergence of Algorithm~\ref{SeekingAlgo}.
However, one insight from the numerical simulations in Example~\ref{ex3} shows that a smaller $\xi$ is required for convergence when dealing with either a larger delay step $\tau$ or a larger number $n$ of agents.
\end{rmk}

\begin{example}\label{ex1}
Consider a five-agent quadratic game where the parameters of each agent's payoff function are given by
$a_{11}^1=-5$, $a_{22}^2=-4,a_{33}^3=-8,a_{44}^4=-2,a_{55}^5=-3,
a_{ij}^1=-1,a_{ij}^2=$ $-0.8,a_{ij}^3=-1.5,a_{ij}^4=-0.4,a_{ij}^5=-0.5,\forall i\neq\!\! j,\mathbf{b}_i=$ $(1,2,3,4,5)^{\top}, i=1,2,3,4,5$.
The NE in this example is calculated by $\s^*=(-360/859,-25/567,-59/378,$ $1007/550,353/241)^\top$ and the agents communicate with each other through a wheel graph as shown in Fig.~\ref{network}, in which agent 1 is the central node.

\begin{figure}[t]
    \centering
    \includegraphics[width=0.2\textwidth]{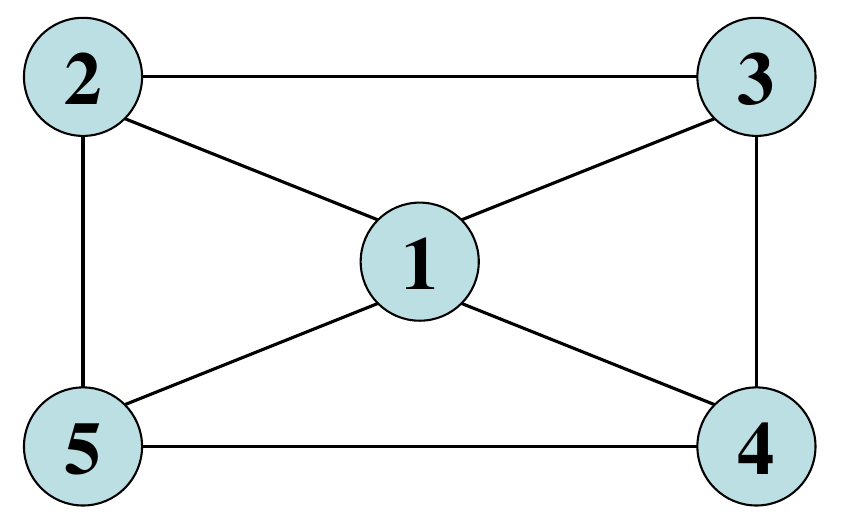}
    \caption{The wheel graph for interaction among agents in Example \ref{ex1}. }
    \label{network}
\end{figure}

Let $\xi=0.08$ be fixed in this example.
With the help of the LMI toolbox in MATLAB,
it can be found out that there exist symmetric positive definite matrices $Q_1$, $Q_2$, and $Q_3$ such that condition \eqref{condi3} in Theorem \ref{thm1} holds until $\tau=3$.
For $\tau=4$, however, the toolbox cannot output proper solutions.
This may imply that the designed seeking dynamics is divergent when $\tau=4$.
We show the strategy profile dynamics for $\tau=3$ and $\tau=4$ in Figs.~\ref{Dysimu2} and \ref{Dysimu3}, respectively, with the initial strategy profile being $\s(0)=(-1,-1,-1,1,1)^{\top}$ and the initial estimations of agents being $\s(t)=\s(0)$, $\hat{\s}^i(t)=\hat{\s}^i(0)$, $i=1,2,3$, for $t=-\tau+1,\ldots,-1$.
From Fig.~\ref{Dysimu2}, it can be observed that all agents' strategies converge to the NE, which is consistent with the conclusion of Theorem~\ref{thm1} since the positive definite matrices satisfying the condition \eqref{condi3} can be found in this case.
This verifies the effectiveness of Theorem \ref{thm1}.
Additionally, it can be observed from Fig.~\ref{Dysimu3} that the evolutions of all agents' strategies are divergent when $\tau=4$.
This suggests that when there exist no positive definite matrices satisfying the condition \eqref{condi3} proposed in Theorem \ref{thm1} under specified learning rate and delay step, the designed seeking dynamics may be divergent in that case.

\begin{figure}[t]
    \centering
    \includegraphics[width=0.4\textwidth]{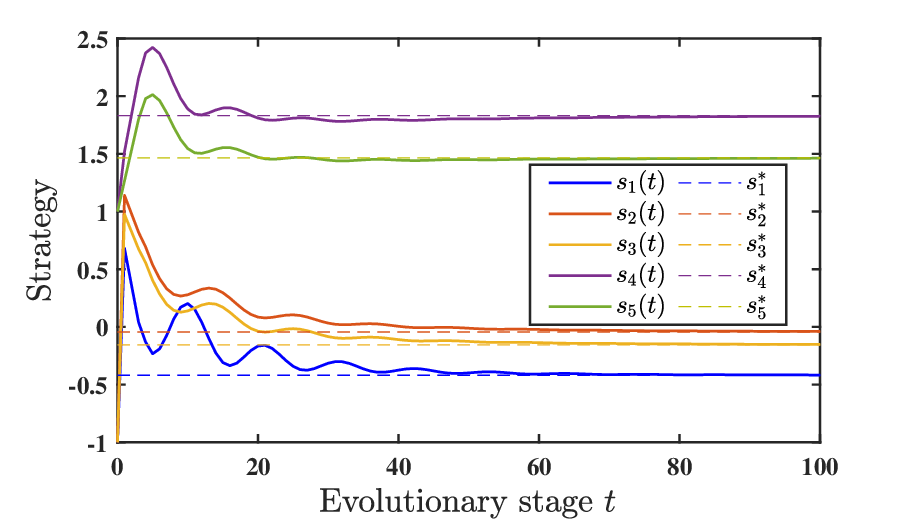}
    \caption{The evolutions of agents' strategies with $\tau=3$ and $\xi=0.08$}
    \label{Dysimu2}
\end{figure}
\begin{figure}[t]
    \centering
    \includegraphics[width=0.4\textwidth]{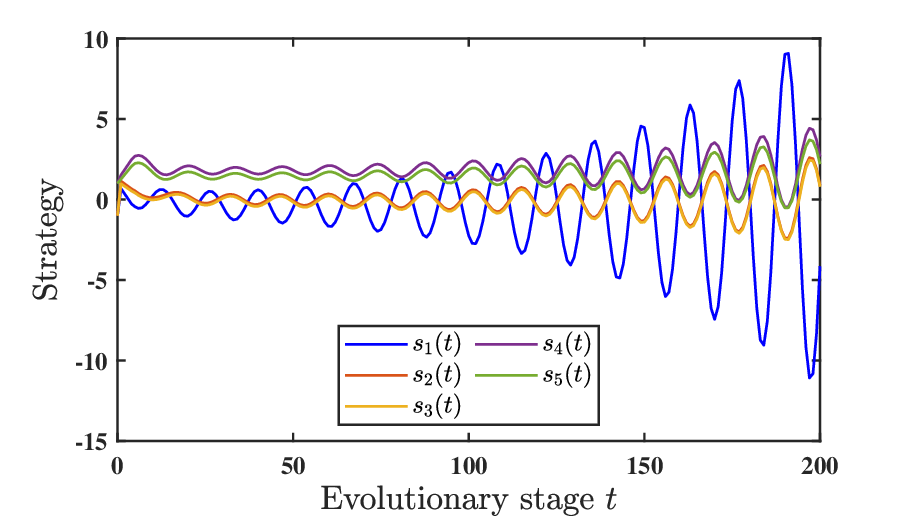}
    \caption{The evolutions of agents' strategies with $\tau=4$ and $\xi=0.08$}
    \label{Dysimu3}
\end{figure}
\end{example}

\subsection{Convergence analysis of Algorithm~\ref{SeekingAlgo} With $\tau=1$}

In this subsection, we consider the second case where all agents exchange exact one-step-delay information with neighbors, i.e., $\tau=1$ in \eqref{estiupdate3}.
In this case, the designed estimation updating dynamics \eqref{estiupdate3} for agent $i$ reduces to
\begin{align}\nonumber\label{estiupdate2}
\hat{s}_{ij}(t)=&\hat{s}_{ij}(t-1)+\xi\Big[\sum\limits_{k\in \mathcal{N}_i} \big(\hat{s}_{kj}(t-1)- \hat{s}_{ij}(t-1)\big)\\
&+\alpha_{ij}(s_j(t-1)-\hat{s}_{ij}(t-1))\Big],\ t\in\mathbb{N}.
\end{align}
Then the updating dynamics of all agents' estimations can be written  in the following compact form
\begin{align}\label{estsystem2}
\hat{\s}(t)=(\I_{n^2}+\xi {P})\hat{\s}(t-1)+\xi B\bar{\s}(t-1),
\end{align}
where $\bar{\s}={\bf 1}_n\otimes \s$, $P$ and $B$ are defined in \eqref{estsystem3}.
Rewrite \eqref{brdynamic} in the following form,
\begin{align}\label{realdynamic}
\bar{\s}(t+1)= {\bf 1}_n\otimes N\hat{\s}(t)+{\bf 1}_n\otimes \mathbf{c}.
\end{align}
Similarly, define the augmented state $\mathfrak{s}(t)\triangleq\col(\bar{\s}(t),\hat{\s}(t))\in\mathbb{R}^{2n^2}$,
from \eqref{realdynamic} and \eqref{estsystem2} one obtains  the following augmented system,
\begin{align}\label{addsystem1}
\mathfrak{s}(t+1)=H(\xi)\mathfrak{s}(t)+\bar{\mathbf{c}},
\end{align}
where $\bar{\mathbf{c}}$ is defined in \eqref{combinedsys} and $H(\xi)\in\mathbb{R}^{2n^2\times 2n^2}$ is given by
\begin{align}\label{eh}
H(\xi)=\left[ {\begin{array}{cc} \0_{n^2\times n^2} & {\bf 1}_n\otimes N\\
\xi B & \I_{n^2}+\xi P \\
\end{array} } \right].
\end{align}

Different from the asymptotical convergence result for the case of $\tau\geq2$ proposed in Theorem~\ref{thm1}, the next result shows that by adopting a proper learning rate $\xi$, all agents' strategies updated according to  Algorithm~\ref{SeekingAlgo} under the case of $\tau=1$ exponentially converge to the NE $\s^*$ of the quadratic game $G$.

\begin{thm}\label{thm2}
Suppose that Assumption \ref{strctilydominant} is satisfied and all agents exchange one-step-delay ($\tau=1$)  information at current stage $t$.
If the learning rate $\xi$ in \eqref{estiupdate2} satisfies
\begin{align}\label{condi2b}
\xi<{\delta}_1, \ \delta_1=\min_{i,j}\{(\alpha_{ij}+\mathcal{L}_{ii})^{-1}\},
\end{align}
then the strategy profile dynamics of all agents under Algorithm~\ref{SeekingAlgo} globally exponentially converge to the NE $\s^*$ of quadratic game $G$.
\end{thm}

{\it Proof.} First, note from Ger\v{s}gorin Theorem (Th. 6.1.1 \cite{Horn2012}) that each eigenvalue $\lambda_i(H(\xi))$, $i=1,\ldots,2n^2$ of matrix $H(\xi)$ satisfies
\begin{align*}
\lambda_i(H(\xi))\in\bigcup\limits_{k=1}^{2n^2}\{z\in\mathbb{C}: |z-[H(\xi)]_{kk}|\leq\sum\limits_{j\neq k}|[H(\xi)]_{kj}|\}.
\end{align*}
We divide the Ger\v{s}gorin disks of the matrix $H(\xi)$ into two parts
\begin{align}
    \mathbf{D}_1=&\bigcup\limits_{k=1}^{n^2}\{z\in\mathbb{C}: |z-[H(\xi)]_{kk}|\leq\sum\limits_{j\neq k}|[H(\xi)]_{kj}|\},\\
    \mathbf{D}_2=&\bigcup\limits_{k=n^2+1}^{2n^2}\{z\in\mathbb{C}: |z-[H(\xi)]_{kk}|\leq\sum\limits_{j\neq k}|[H(\xi)]_{kj}|\}.
\end{align}
Since the first $n^2$ rows of $H(\xi)$ is $\left[ \0_{n^2\times n^2} , {\bf 1}_n\otimes N  \right]$,
there exist only $n$ sets of different disks with $[H(\xi)]_{kk}=0$, $k=1,\ldots,n^2$ and
\begin{align}
    \mathbf{D}_1=\bigcup\limits_{k=1}^{n}\{z\in\mathbb{C}: |z-0|\leq\sum\limits_{j\neq k}|[N]_{kj}|\}.
\end{align}
Since  the remaining $n^2$ rows of $H(\xi)$ is $[\xi B,\I+\xi P]$,  $[H(\xi)]_{kk}=[\I+\xi P]_{kk}$ for $k=n^2+1,\ldots,2n^2$ and $\mathbf{D}_2$ can be rewritten by
\begin{align*}
    \mathbf{D}_2=&\bigcup\limits_{i,j=1}^n\!\{z\in\mathbb{C}: |z-(1-\xi(\mathcal{L}_{ii}+\alpha_{ij}))|\leq\xi(\mathcal{L}_{ii}+\alpha_{ij})\}.
\end{align*}
For each eigenvalue $\lambda_i(H(\xi))$, if $\lambda_i(H(\xi))\in\mathbf{D}_1$, then there exists at least one $i_1\in\{1,2,\ldots,n\}$ such that
\begin{align}\label{sr1}
|\lambda_i(H(\xi))-0|\leq \sum\limits_{j\neq i_1}|[N]_{i_1j}|=\sum\limits_{j\neq i_1}{|a_{i_1j}^{i_1}|}/{|a_{i_1i_1}^{i_1}|}<1
\end{align}
under Assumption \ref{strctilydominant} and hence $|\lambda_i(H(\xi))|<1$.

If $\lambda_i(H(\xi))\in\mathbf{D}_2$,
then there exists at least one $i_2\in\{1,2,\ldots,n\}$ and one $j_2\in\{1,2,\ldots,n\}$ such that
\begin{align}\label{sr2}
|\lambda_i(H(\xi))-h_i|\leq1-h_i,
\end{align}
where $h_i\triangleq1-\xi(\mathcal{L}_{i_2i_2}+\alpha_{i_2j_2})$ satisfies $h_i\in(0,1)$ under \eqref{condi2b}.
The above inequality \eqref{sr2} implies that $\lambda_i(H(\xi))$ is located in the Ger\v{s}gorin disk whose center is $(h_i,0)^\top$ and radius is $1-h_i$.
It is worth noting that for any $h_i\in(0,1)$, the disk is always inscribed in the unit circle at  point $(1,0)^\top$ in the complex plane.
See Fig.~\ref{spectralradius} for illustrations of the above Ger\v{s}gorin disk and the unit circle.

Let $\lambda_i(H(\xi))$ be represented by $\lambda_i(H(\xi))=a_{i}+b_{i}\ii$, where $\ii$ is the imaginary unit, $a_i\in[2h_i-1,1]$ and $b_i\in[h_i-1,$ $1-h_i]$.
Note that the point $\zeta=(a_i,\sqrt{(1-h_i)^2-(a_i-h_i)^2})^\top$ is on the circle $\{(x,y)^\top\in\mathbb{R}^2:|(x-h_i)+y\ii|=1-h_i\}$ and $\zeta$ is with the same magnitude of Re-axis as $\lambda_i(H(\xi)$.
Therefore, one has
\begin{align}\nonumber
|\lambda_i(H(\xi))|&\leq|\zeta|=\sqrt{a_{i}^2+(1-h_i)^2-(a_i-h_i)^2}\\
&=\sqrt{1+2(a_i-1)h_i}\leq1,
\end{align}
where the first equality holds if and only if  $\lambda_i(H(\xi))=\zeta$ and the second equality holds if and only if $a_i=1$ and $b_i=0$.
This implies $|\lambda_i(H(\xi))|=1$ if and only if $\lambda_i(H(\xi))=1$.
For any other feasible value of $\lambda_i(H(\xi))$, one has $|\lambda_i(H(\xi))|<1$.

\begin{figure}[t]
    \centering
    \includegraphics[width=0.35\textwidth]{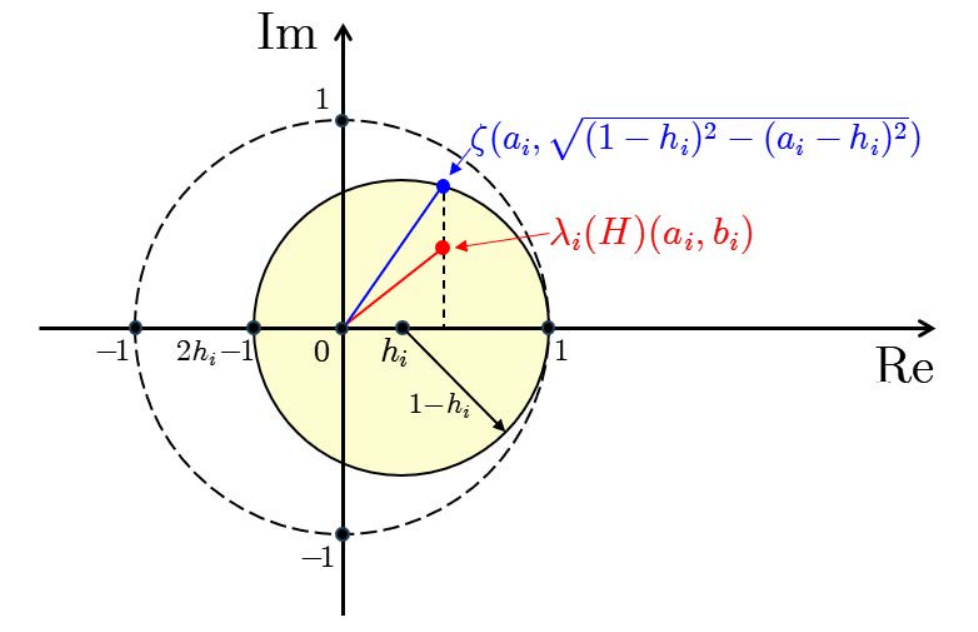}
    \caption{The illustrations of unit circle and the (yellow) Ger\v{s}gorin disk whose center is $(h_i,0)^\top$ and radius is $1-h_i$.}
    \label{spectralradius}
\end{figure}

Next we prove that the real number 1 is not an eigenvalue of the matrix $H(\xi)$ by contradiction.
Suppose that 1 is an eigenvalue of $H(\xi)$ whose corresponding eigenvector is denoted by $\vv=\col(\vv_1,\vv_{2})\in\mathbb{R}^{2n^2}$.
The components $\vv_i, i=1,2$ of $\vv$ are all $n^2$-dimensional column vectors satisfying that
\begin{align*}
\left[ {\begin{array}{cc} \0_{n^2\times n^2} & {\bf 1}_n\otimes N\\
\xi B & \I_{n^2}+\xi P \\
\end{array} } \right]\vv=1\cdot\vv.
\end{align*}
The above equation can be decomposed into
\begin{subequations}
\begin{align}\label{eigencondi1}
&\vv_1=({\bf 1}_n\otimes N)\vv_{2}, \\ \label{eigencondi2}
& \xi B\vv_{1}=-\xi P\vv_2=\xi(\mathcal{L}\otimes \I_n+B)\vv_{2}.
\end{align}
\end{subequations}
From \eqref{eigencondi1} and \eqref{eigencondi2} with $\xi>0$, one obtains
\begin{align*}
    B({\bf 1}_n\otimes N)\vv_{2}=& B\vv_{1}=(\mathcal{L}\otimes \I_n+B)\vv_{2},
\end{align*}
which is equal to
\begin{align}\label{homoequation}
    \big(B({\bf 1}_n\otimes N-\I_{n^2})-\mathcal{L}\otimes \I_n\big)\vv_{2}={\bf 0}_{n^2}.
\end{align}
Since the matrix  $B({\1}_n\otimes N)-B-\mathcal{L}\otimes \I_n$ as mentioned in \eqref{NonsingularMatrix} is nonsingular, the above homogeneous linear equation \eqref{homoequation} only has zero solution, i.e., $\vv_{2}={\bf 0}_{n^2}$.
Furthermore, from \eqref{eigencondi1}, one has
\begin{align}\label{allzero}
\vv_{1}=({\bf 1}_n\otimes N)\vv_{2}={\bf 0}_{n^2}.
\end{align}
Finally, it arrives at $\vv={\bf 0}_{2n^2}$, which contradicts the assumption that $\vv$ is an eigenvector.
Therefore, 1 is not an eigenvalue of $H$ and $|\lambda_i(H(\xi))|<1$ holds for any $i=1,2,\ldots,2n^2$, which implies $\rho(H(\xi))<1$.
According to Theorem 22.11 in \cite{Rugh1996}, the system \eqref{addsystem1} is exponentially stable, i.e., there exists $\beta\geq1$ and $\gamma\in(0,1)$ such that $\|\mathfrak{s}(t)\|_2\leq\beta\gamma^t\|\mathfrak{s}(0)\|_2$ with $\|\cdot\|_2$ being the Euclidean norm.

Denote the equilibrium state of system \eqref{addsystem1} by $\mathfrak{s}^e=\col(\bar{\s}^e,\hat{\s}^e)$. 
According to \eqref{addsystem1}, one has
\begin{align}
\mathfrak{s}^e=H(\xi)\mathfrak{s}^e +\bar{\mathbf{c}},
\end{align}
which is equivalent to \eqref{equilibriumstate} since $H(\xi)=H_1+H_2$ according to \eqref{H1H2} and \eqref{eh}.
Then the rest of the proof follows the proof of Theorem~\ref{thm1} and is omitted here.$\hfill\Box$

\begin{rmk}
The authors of \cite{Li2025} designed an asynchronous seeking algorithm that utilizes both a communication graph and an additional interference graph to ensure almost sure convergence to the NE under information delay.
In contrast, the algorithm proposed in this paper is synchronous, relies solely on the communication graph, and achieves deterministic convergence.
Additionally, unlike \cite{LiuJ2024} where the designed seeking algorithm sub-linearly converges to the NE, our designed  algorithm can further achieve an exponential convergence result in the case of $\tau=1$.
\end{rmk}

It should be pointed out that the above upper bound $\delta_1$ of learning rate $\xi$ for ensuring the convergence of the designed seeking dynamics is relatively conservative.
Using some learning rates larger than $\delta_1$ may also steer the seeking dynamics to the NE, please see   Fig.~\ref{learningratetau1} in Example \ref{ex2}.
However,  when the learning rate $\xi$ exceeds a bound, the seeking dynamics become unstable, as stated in the following.

\begin{thm}\label{thm3}
Suppose that Assumption \ref{strctilydominant} is satisfied and all agents  exchange one-step-delay ($\tau=1$)  information at current stage $t$.
If the learning rate $\xi$ in \eqref{estiupdate2} satisfies
\begin{align}\label{condi2c}
\xi>\delta_2, \ \delta_2=\frac{3n^2}{(n+1)\mathbf{tr}(\mathcal{L})},
\end{align}
where $\mathbf{tr}(\mathcal{L})$ is the trace of the Laplacian matrix $\mathcal{L}$ of communication graph $(\mathcal{N},\mathcal{E})$, then the strategy profile dynamics of all agents under Algorithm~\ref{SeekingAlgo} becomes unstable.
\end{thm}

\emph{Proof}.
Let $\mathbf{tr}(H(\xi))$ denote the trace of matrix $H(\xi)$, which is exactly the sum of all eigenvalues of $H(\xi)$, i.e.,
\begin{align}\!\label{traceH}
\mathbf{tr}(H(\xi))=\sum\nolimits_{i,j\in\mathcal{N}} 1-\xi(\mathcal{L}_{ii}+\alpha_{ij})=\sum\nolimits_{i=1}^{2n^2}\lambda_i(H(\xi)).
\end{align}
According to the definition of matrix $B$ in \eqref{Bdefi}, one has
\begin{align}\label{Bstructure}
    \sum\nolimits_{j=1,j\neq i}^n \alpha_{ij}=\mathcal{L}_{ii},\ \alpha_{ii}=0,\ i=1,\ldots,n.
\end{align}
Substituting \eqref{Bstructure} into \eqref{traceH}, one has
\begin{align*}
    \mathbf{tr}(H(\xi))&\!=\!n^2\!-\!\big(\sum\limits_{i,j\in\mathcal{N}} \mathcal{L}_{ii}+\sum\limits_{i,j\in\mathcal{N}}\alpha_{ij}\big)\xi\\
    &\!=\!n^2\!-\!\big(n\!\sum\limits_{i\in\mathcal{N}} \mathcal{L}_{ii}+\!\sum\limits_{i\in\mathcal{N}}\mathcal{L}_{ii}\big)\xi\!=\!n^2\!-\!\xi(n+1)\mathbf{tr}(\mathcal{L}).
\end{align*}
Since $\xi$ satisfies condition \eqref{condi2c},  it holds $\xi(n+1)\mathbf{tr}(\mathcal{L})>3n^2$, which implies $\mathbf{tr}(H(\xi))<-2n^2$.
Denote $\bar{\lambda}(H(\xi))\triangleq(1/2n^2)\mathbf{tr}(H(\xi))$ as the mean value of eigenvalues of $H(\xi)$, one further has
\begin{align}
\bar{\lambda}(H(\xi))<-1.
\end{align}
The above inequality implies that there exists at least one eigenvalue $\lambda_{i}(H(\xi))$ such that $\lambda_{i}(H(\xi))<-1$.
As a result, the spectral radius $\rho(H(\xi))$ of matrix $H(\xi)$ will larger than 1.
According to Theorem 22.11 in \cite{Rugh1996},
$\rho(H(\xi))>1$ leads the system \eqref{addsystem1} to be unstable, which further results in the instability of the designed strategy profile dynamics \eqref{brdynamic}.
The proof is completed. $\hfill\Box$

\begin{table}[t]
  \caption{Value of $\delta_1$ and $\delta_2$ in different graphs with $n$ nodes}
  \label{LearningValues}
  \centering
  \setlength{\tabcolsep}{1.2mm}{
  \begin{tabular}{c|c|c|c|c}
  \midrule
  \specialrule{0em}{0.9pt}{0.9pt}\hline\rule{0pt}{6pt}
               & Ring graph     & Complete graph  & Star graph      & Wheel graph\\[1pt]  \hline
  {$\delta_1$} & $1/3$          & $1/n$           & $1/n$           & $1/n$ \\[1pt]  \hline
  {$\delta_2$} & $3n/2(n+1)$    & $3n/(n^2-1)$    & $3n^2/2(n^2-1)$ & $3n^2/4(n^2-1)$\\[1pt]  \hline
  \midrule
  \specialrule{0em}{0.3pt}{0.3pt}
  \end{tabular}}
\end{table}

\begin{rmk}
When the communication graph of game $G$ is given, one can calculate the values of the  bounds $\delta_1$ and $\delta_2$ in Theorems~\ref{thm2} and \ref{thm3} according to \eqref{condi2b} and \eqref{condi2c}, respectively.
For instance,  four typical communication topologies (including ring graph, complete graph, star graph, and wheel graph) are considered in this paper and the values of bounds $\delta_1$ and $\delta_2$ are shown in Table~\ref{LearningValues}.
It can be observed from  Table~\ref{LearningValues} that using the ring graph as the communication topology yields  a larger
guaranteed bound of convergence on the learning rate  $\xi$ than  the other three graphs especially for the case when the number $n$ of agents is large.
This is because the guaranteed bound of convergence of the ring graph is always given by $1/3$ no matter how many agents exist in the game.
Another point worth emphasizing is that the bound $\delta_2$ of  learning rate   proposed in Theorem~\ref{thm3} is closely related to the structure of communication graph.
It can be observed from \eqref{condi2c} that a denser graph yielding a smaller $\delta_2$.
\end{rmk}

It should be emphasized that the proposed instability bound $\delta_2$ is applicable to all communication topologies, and may appear relatively conservative for special types of topologies.
To illustrate this, two topologies, namely, the ring graph and the complete graph, are considered for numerical sensitivity analysis.
The maximal learning rate $\xi_{\max}$ (which ensures the algorithm's convergence) for these two topologies with different numbers of agents are presented in Table~\ref{Ring}.
From Table~\ref{Ring}, it can be observed that as  $n$ increases, $\delta_2$ for the ring graph tends to 1.5, which deviates from $\xi_{\max}$ that tends to a constant value of 0.455.
In contrast, the bound $\delta_2$ for the complete graph approaches $\xi_{\max}$ as $n$ increases.
This indicates that the proposed bound is somewhat satisfactory in the case of the complete graph.
    \begin{table}[t]
  \caption{The $\xi_{\max}$ in two graphs with different number of agents}
  \label{Ring}
  \centering
  \setlength{\tabcolsep}{1.2mm}{
  \begin{tabular}{c|c|c|c|c|c|c|c}
  \midrule
  \specialrule{0em}{0.9pt}{0.9pt}\hline\rule{0pt}{6pt}
   &  $n$           & 3     & 5     & 8     & 10    & 20    & 50    \\[1pt]  \hline
    \multirow{2}{*}{Ring graph}         & {$\delta_2$}    & 1.125 & 1.250 & 1.333 & 1.363 & 1.429 & 1.471      \\[1pt]  \cline{2-8}
             & {$\xi_{\max}$}  & 0.570 & 0.472 & 0.453 & 0.454 & 0.455 & 0.455 \\[1pt]  \hline
    \multirow{2}{*}{Complete graph} & {$\delta_2$}    & 1.125 & 0.625 & 0.381 & 0.303 & 0.150 & 0.06   \\[1pt]  \cline{2-8}
     & {$\xi_{\max}$}  & 0.570 & 0.374 & 0.243 & 0.200 & 0.100 & 0.04 \\[1pt]  \hline
  \midrule
  \specialrule{0em}{0.3pt}{0.3pt}
  \end{tabular}}
  \end{table}

\begin{example}\label{ex2}
\begin{figure}[t]
    \centering
    \includegraphics[width=0.4\textwidth]{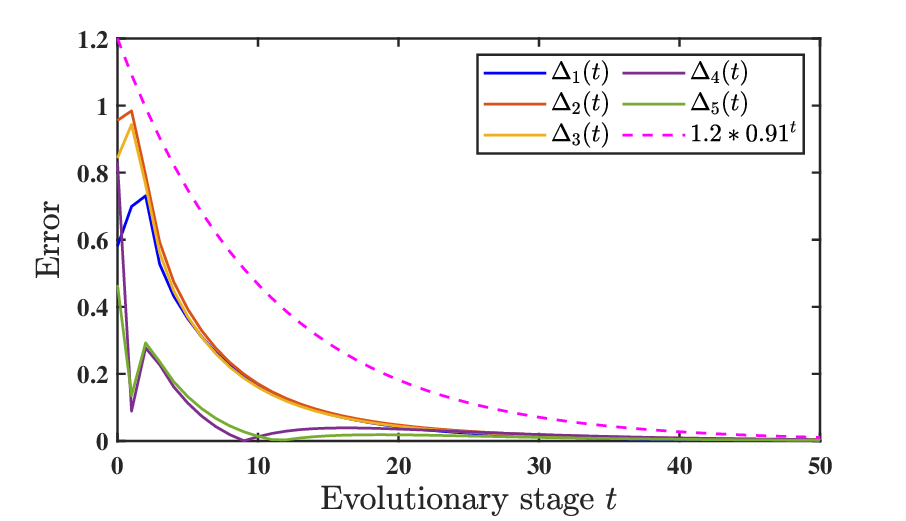}
    \caption{ The evolutions of agents' errors under the case of $\tau=1$}
    \label{Dysimu}
    \vspace{-10pt}
\end{figure}
Consider the example shown in Example~\ref{ex1} with five agents, here we verify the proposed results in Theorems~\ref{thm2} and \ref{thm3} for $\tau=1$.
Let the initial strategy profile the same as in Example~\ref{ex1} and the estimations of agents be respectively  $\hat{\s}^i(-1)=-\1_5$, $i=1,2,3$, $\hat{\s}^j(-1)=\1_5$, $j=4,5$.
Denote $\Delta_i(t)\triangleq\|s_i(t)-s_i^*\|_2$ to characterize the error between agent $i$'s strategy $s_i(t)$ and  the strategy $s_i^*$ in NE $\s^*$.
Here, $\|\cdot\|_2$ represents the Euclidean norm.
According to Table~\ref{LearningValues}, one has $\delta_1=0.2$ in \eqref{condi2b} of Theorem~\ref{thm2} since $n=5$ in this example.
Let $\xi=0.18<\delta_1$, all agents' error evolutions under the strategy updating dynamics \eqref{strategydyna} along with \eqref{estiupdate2} are shown in Fig.~\ref{Dysimu}.
From Fig.~\ref{Dysimu}, it can be observed that each agent's strategy under the designed seeking dynamics converges to the NE $\s^*$ as the error $\Delta_i(t)$ converges to zero being suppressed exponentially, which verifies Theorem \ref{thm2}.

Next we show the influence of learning rate $\xi$ on the convergence of the designed seeking dynamics in the case of $\tau=1$.
Let the termination condition be $\Delta=0.0001$ for strategy profile dynamics \eqref{brdynamic} along with estimation updating dynamics \eqref{estiupdate2}.
The maximal updating stage is  5000 and we define the terminal stage $\mathbf{T}_\Delta$  by
\begin{align}\label{terminalstage}
\mathbf{T}_\Delta=\arg\min\limits_{T\leq 5000}\{T:\|\s(T)-\s^*\|_2\leq\Delta\}.
\end{align}
Let the learning rate $\xi$ increase from 0.05 to 1, the corresponding terminal stages under different $\xi$ are shown in Fig.~\ref{learningratetau1}, from which it can be observed that the curve of terminal stages $\mathbf{T}$ first decreases and then increases as $\xi$ increases.
Note that when $\xi>1/3$, the designed seeking dynamics is unstable by calculating the spectral radius $\rho(H(\xi))>1$ for $\xi>1/3$, hence the curve rises sharply and we use the dashed line to indicate that even if it reaches the maximal updating stage, the strategy profile error fails to satisfy the terminal condition $\Delta$.
According to Table~\ref{LearningValues}, since agents exchange information on the wheel graph with $n=5$ in this example, one calculates $\delta_2=3n^2/4(n^2-1)=25/32>1/3$, thus the seeking dynamics is certainly unstable when $\xi>\delta_2$, which verifies Theorem~\ref{thm3}.

Moreover, it can be observed from Fig.~\ref{learningratetau1} that the designed seeking dynamics also converges when using the learning rate $\xi=0.08$ for $\tau=1$.
Compared with the divergent dynamics shown in Fig.~\ref{Dysimu3} where agents exchange 4-step-delay (i.e., $\tau=4$) information with the same learning rate $\xi=0.08$, it can be found that the learning rate should be reduced to ensure the convergence of the strategy update dynamics when agents exchange information with longer delays.

\begin{figure}[t]
    \centering
    \includegraphics[width=0.4\textwidth]{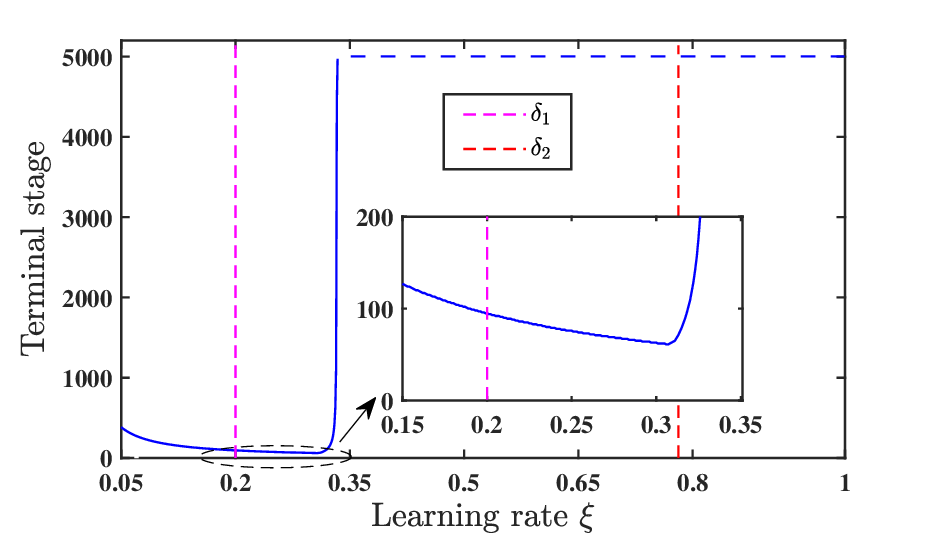}
    \caption{The terminal stages with different learning rates $\xi$}
    \label{learningratetau1}
\end{figure}
\end{example}
\begin{figure}[t]
    \centering
    \includegraphics[width=0.4\textwidth]{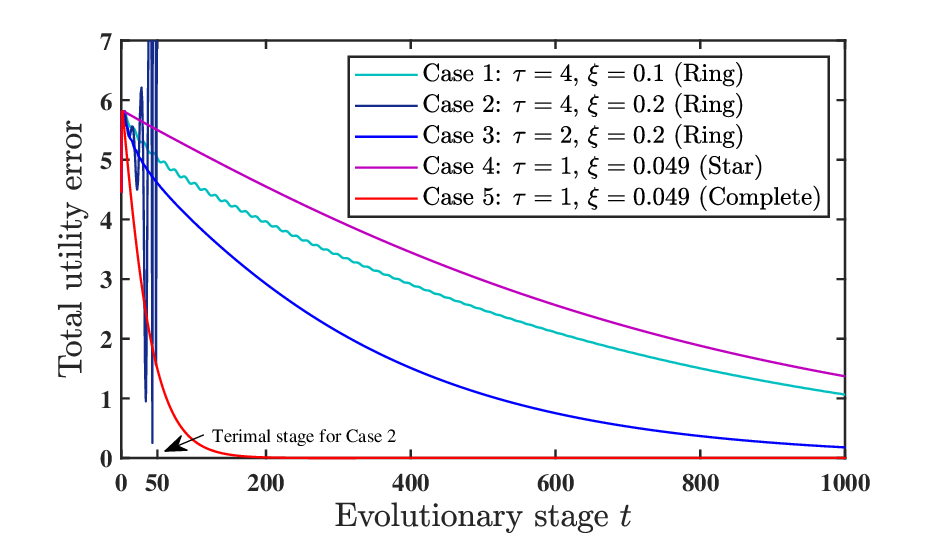}
    \caption{The evolutions of 20 agents' total utility errors }
    \label{Errordynamics}
    \vspace{-10pt}
\end{figure}
\begin{example}\label{ex3}
To validate the scalability of Algorithm~\ref{SeekingAlgo} for large networks, a more complex scenario with 20 agents is provided, where $a_{ii}^i=[\mathbf{a}]_i$, $i=1,2,\ldots,20$,  $\mathbf{a}=-(43,27,35,41,47,49,36,24,24$, $27,46,27,45,27,48,30,26,27,39,34)$, $a_{ij}^i=1$ for $i\neq j$, $\mathbf{b}_i=(1,2,\ldots,$ $20)^\top$ for $i=1,2,\ldots,20$.
By calculation, the total utilities of agents at the NE $\s^*$ of this game is -4.4563.
Five cases of parameters and topologies are considered in this example and the evolutions of agents' total utility errors $|\sum_{i\in \mathcal{N}}\big(J_i(\s(t))-J_i(\s^*)\big)|$ under these cases are shown in Fig.~\ref{Errordynamics}.
From Fig.~\ref{Errordynamics} it can be observed that except for \emph{Case~2} of $\tau=4$, $\xi=0.2$ in a ring graph (which violates the condition in Theorem~\ref{thm1}), the strategy profile dynamics in other cases can all converge to the NE of this quadratic game.
Since the strategy profile dynamics in \emph{Case~2} are divergent, we set the terminal stage to be 50 in this case and omit the remaining trajectory of  agents' total utility error for $t>50$.

Additionally, to provide some insight into setting a `good' learning rate that ensures the convergence of the proposed algorithm, we show the maximal learning rate $\xi_{\max}$ that guarantees the convergence of the proposed algorithm under two communication graphs with different delays and different number of agents in Figure~\ref{MLR}.
From Figure~\ref{MLR}(a) (resp. Figure~\ref{MLR}(b)), it can be observed that in both graphs, when the number of agents (resp. the delay step) is fixed, to ensure the convergence of the proposed algorithm, there is an inverse proportional relationship between the delay step $\tau$ (resp. the number $n$ of agents) and the learning rate $\xi$.

  \begin{figure}[t]\centering
  \subfloat{\centering
  \begin{minipage}[t]{0.5\linewidth}\centering
      \includegraphics[width=4.6cm]{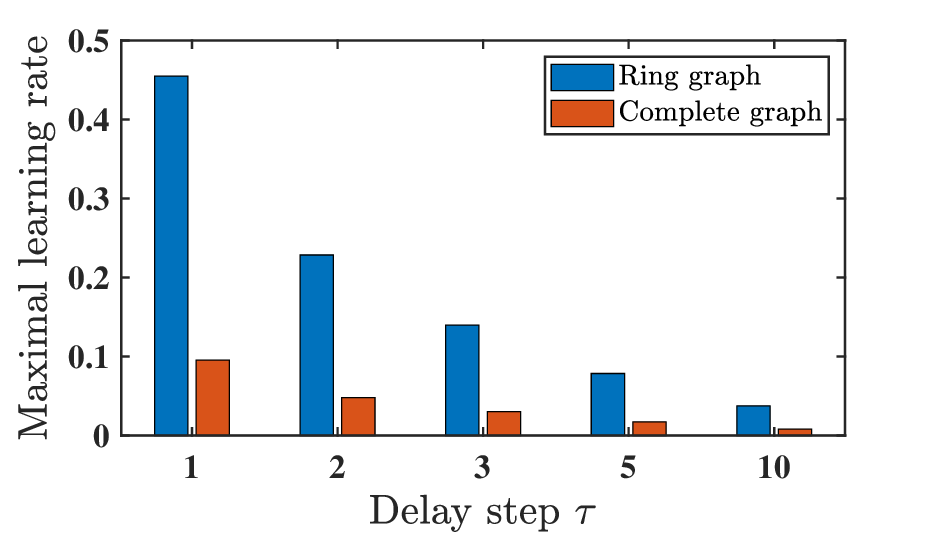}
  \centerline{\small (a)~Different $\tau$ with 20 agents}
  \end{minipage}
  \begin{minipage}[t]{0.5\linewidth}\centering
      \includegraphics[width=4.6cm]{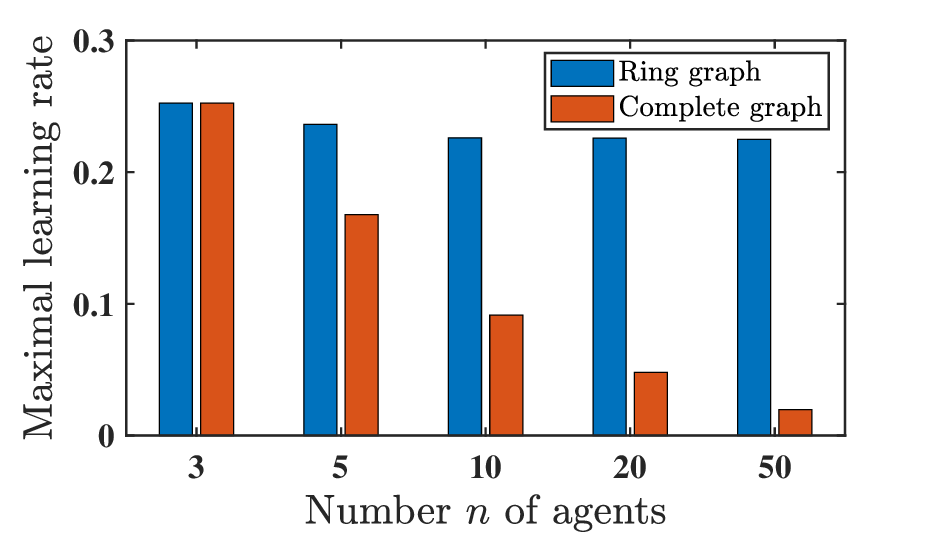}
  \centerline{\small (b)~Different $n$  with $\tau=2$}
  \end{minipage}}
  \caption{The maximal learning rate $\xi_{\max}$ that ensures the convergence of Algorithm~1 under two cases}
  \label{MLR}
  \vspace{-10pt}
  \end{figure}
\end{example}

\section{Conclusion}

In this paper, the seeking of NE in a non-cooperative quadratic game with selfish agents was investigated, where agents exchange both delayed strategy and estimation information with their neighbors.
Based on the delayed information, an estimation updating mechanism was designed for the best-response-based strategy updating dynamics.
When agents exchange multi-step-delay information, a sufficient condition on both the learning rate and the delay step was proposed to ensure the asymptotical convergence of 
 NE of the quadratic game.
In the case of exchanging one-step-delay information, a sufficient condition on the learning rate was proposed to achieve the exponential convergence.
Another condition on the learning rate was provided to prevent the divergence of the strategy updating dynamics in this case. 
The obtained results showed that even agents exchange delayed information with neighbors, using a proper estimation updating dynamics with an appropriate learning rate can still promote the convergence of agents' strategies to NE.
Numerical simulations were proposed to verify the developed theoretical results.

Future works may involve providing the existence condition of matrices $Q_1$, $Q_2$, and $Q_3$ in Theorem~\ref{thm1}, extending the payoff function to a more general setting, analyzing the influence of the delay step and learning rate on the convergence speed, and considering more realistic settings such as noise, packet dropout, or asynchronous updates.

\begin{thebibliography}{00}
\bibitem{Varga2023} B. Varga, J. Inga, and S. Hohmann, ``Limited information shared control: a potential game approach,'' IEEE T. Hum.-Mach. Syst., vol, 53, no. 2, pp. 282-292, 2023.
\bibitem{Varga2022} L. Z. Varga, ``Solutions to the routing problem: towards trustworthy auto- nomous vehicles,'' Artif. Intell. Rev., vol. 55, no. 7, pp. 5445-5484, 2022.
\bibitem{Liu2024} D. Liu, \emph{et~al}, ``Game of drones: Intelligent online decision making of multi-UAV confrontation,'' IEEE Trans. Emerg. Top. Comput. Intell., vol. 8, no. 2, pp. 2086-2100, 2024.
\bibitem{Frihauf2012} P. Frihauf, M. Krstic and T. Ba\c{s}ar, ``Nash equilibrium seeking in noncooperative games,'' IEEE Trans. Autom. Control, vol. 57, no. 5, pp. 1192-1207, 2012.
\bibitem{Martirosyan2024} E. Martirosyan, and M. Cao, ``Reinforcement learning for inverse linear-quadratic dynamic non-cooperative games,'' Syst. Control Lett., vol. 191, pp. 105883, 2024.
\bibitem{Xu2024} Y. Xu, M. Xiao, Y. Zhu, J. Wu, S. Zhang and J. Zhou, ``AoI-guaranteed incentive mechanism for mobile crowdsensing with freshness concerns,'' IEEE. Trans. Mob. Comput., vol. 23, no. 5, pp. 4107-4125, 2024.
\bibitem{Li2021} S. H. Q. Li, L. Ratliff and B. A\c{c}{\i}kme\c{s}e, ``Disturbance decoupling for gradient-based multi-agent learning with quadratic costs,'' IEEE Control Syst. Lett., vol. 5, no. 1, pp. 223-228, 2021.
\bibitem{Basar1997} T. Ba\c{s}ar and G. J. Olsder, \emph{Dynamic Noncooperative Game Theory}, 2nd ed. Philadelphia, PA: SIAM, 1999.
\bibitem{Nash1951} J. Nash, ``Non-cooperative game,''  Ann. Math., vol. 54, no. 2, pp. 286--295, 1951.
\bibitem{Monderer1996} D. Monderer and L. S. Shapley, ``Potential games,'' Games Econ. Behav.,'' vol. 14, no. 1, pp. 124-143, 1996.
\bibitem{Liu2019} T. Liu, J. Wang, X. Zhang and D. Cheng, ``Game theoretic control of multiagent systems,'' SIAM J. Control Optim., vol. 57, no. 3, pp. 1691-1709, 2019.
%


\bibitem{Yan2024a} Y. Yan and T. Hayakawa, ``Pareto-improving incentive mechanism for noncooperative dynamical systems under sustainable budget constraint,'' IEEE Trans. Autom. Control, vol. 69, no. 7, pp. 4291-4306, 2024.
\bibitem{Yan2024b} Y. Yan and T. Hayakawa, ``Incorporation of likely future actions of agents into pseudo-gradient dynamics of noncooperative games,'' IEEE Trans. Autom. Control, vol. 69, no. 11, pp. 7662-7677, 2024.
\bibitem{Ye2017} M. Ye, and G. Hu, ``Distributed Nash equilibrium seeking by a consensus based approach,'' IEEE Trans. Autom. Control, vol. 62, no. 9, pp. 4811--4818,  2017.
\bibitem{Wang2025} M. Wang, Y. Wu and S. Qin, ``Generalized Nash equilibrium seeking for noncooperative game with different monotonicities by adaptive neurodynamic algorithm,'' IEEE Trans. Neural Netw. Learn. Syst., vol. 36, no. 4, pp. 7637-7650, 2025.
\bibitem{Marden2009} J. R. Marden, G. Arslan, and J. S. Shamma, ``Joint strategy fictitious play with inertia for potential games,'' IEEE Trans. Autom. Control, vol. 54, no. 2, pp. 208--220, Feb. 2009.
\bibitem{Belgioioso2023} G. Belgioioso, and S. Grammatico, ``Semi-decentralized generalized Nash equilibrium seeking in monotone  aggregative games,'' IEEE Trans. Autom. Control, vol. 68, no. 1, pp.140-155, 2023.
\bibitem{Pavel2020} L. Pavel, ``Distributed GNE seeking under partial-decision information over networks via a doubly-augmented operator splitting approach,'' IEEE Trans. Autom. Control, vol. 67, no. 4, pp. 1584-1597,  2020.
\bibitem{Shamma2005} J. S. Shamma, and G. Arslan, ``Dynamic fictitious play, dynamic gradient play, and distributed convergence to Nash equilibria,'' IEEE Trans. Autom. Control, vol. 50, no. 3, pp. 312-327,  2005.
\bibitem{Govaert2021} A. Govaert, P. Ramazi and M. Cao, ``Rationality, imitation, and rational imitation in spatial public goods games,''  IEEE Trans. Control Netw. Syst., vol. 8, no. 3, pp. 1324--1335,  2021.
\bibitem{Hart2006} S. Hart, and A. Mas-Colell, ``Stochastic uncoupled dynamics and Nash equilibrium,'' Games Econ. Behav., vol. 57, no. 2, pp. 286-303, 2006.
\bibitem{Lou2016} Y. Lou, Y. Hong, L. Xie, G. Shi and K. H. Johansson, ``Nash equilibrium computation in subnetwork zero-sum games with switching communications,'' IEEE Trans. Autom. Control, vol. 61, no. 10, pp. 2920-2935,  2016.
\bibitem{Fang2022} X. Fang, G. Wen, J. Zhou, J. L\"{u}, and G. Chen, ``Distributed Nash equilibrium seeking for aggregative games with directed communication graphs,'' IEEE Trans. Circuits Syst. I-Regul. Pap., vol. 69, no. 8, pp. 3339-3353, 2022.
\bibitem{Xu2023} W. Xu, Z. Wang, G. Hu, and J. K\"{u}rths, ``Hybrid Nash equilibrium seeking under partial-decision information: an adaptive dynamic event-triggered approach,'' IEEE Trans. Autom. Control, vol. 68, no. 10, pp. 5862-5876,  2023.
\bibitem{Huang2023} S. J. Huang, J. L. Lei, and Y. G. Hong, ``A Linearly convergent distributed Nash equilibrium seeking algorithm for aggregative games,'' IEEE Trans. Autom. Control, vol. 68, no. 3, pp. 1753-1759,  2023.
\bibitem{Cramton1991} P. C. Cramton, ``Dynamic bargaining with transaction costs,'' Manage. Sci., vol. 37, no. 10, pp. 1221-1233, 1991.
\bibitem{Spector2022} D. Spector, ``Cheap talk, monitoring and collusion,'' Review of Industrial Organization, vol. 60, pp. 193--216, 2022.
\bibitem{Wang2022} X. Wang, X. Sun, M. Ye and K. Liu, ``Robust distributed Nash equilibrium seeking for games under attacks and communication delays,`` IEEE Trans. Autom. Control, vol. 67, no. 9, pp. 4892-4899, 2022.
\bibitem{Ai2020} X. Ai, ``Distributed Nash equilibrium seeking for networked games of multiple high-order systems with disturbance rejection and communication delay,'' Nonlinear Dyn., vol. 101, pp. 961-976, 2020.
\bibitem{Li2025} H. Li, $et~al.$, ``Convergence analysis of distributed generalized Nash equilibria seeking algorithm with asynchrony and delays,'' IEEE Trans. Autom. Control, vol. 70, no. 1, pp. 642-648,  2025.
\bibitem{LiuJ2024} J. Liu, L. Li, W. Xu, and D. W. C. Ho, ``Nash Equilibrium seeking over time-varying networks with time-varying delays,'' Proc. of European Control Conference, pp. 792-797, 2024.
\bibitem{Nedic2010} A. Nedic and A. Ozdaglar, ``Convergence rate for consensus with delays,''
J. Glob. Optim., vol. 47, pp. 437-456, 2010.
\bibitem{Chen2024} G. Chen and Y. Zhou, ``Dynamic estimation over distributed sensing network with communication delays,'' IEEE Trans. Ind. Inform., vol. 20, no. 4, pp. 5449-5458, 2024.
\bibitem{Seuret2015} A. Seuret, F. Gouaisbaut, and E. Fridman, ``Stability of discrete-Time systems with time-varying delays via a novel summation inequality,'' IEEE Trans. Autom. Control, vol. 60, no. 10, pp. 2740-2745, 2015.
\bibitem{Yan2022} Y. Yan, and T. Hayakawa, ``Stability analysis of Nash equilibrium for two-agent loss-aversion-based noncooperative switched systems,'' IEEE Trans. Autom. Control, vol. 67, no. 5, pp. 2505-2513, 2022.
\bibitem{Rugh1996} W. J. Rugh, \emph{Linear System Theory}, 2nd ed., Prentice Hall, Upper Saddle River, NJ, 1996.
\bibitem{Horn2012} R. A. Horn, and C. R. Johnson, ``Matrix Analysis,'' second editon, Cambridge University Press, 2012.

%
%
%



\end{thebibliography}
\end{document}